\documentclass[fleqn,10pt,twocolumn]{SICE_ISCS}

\usepackage{graphicx}      % include this line if your document contains figures
\usepackage{amsmath} 
\usepackage{amssymb} 
\usepackage{booktabs}

\usepackage{amsthm}

\DeclareMathOperator*{\argmax}{arg\,max}

\usepackage{algorithm}
\usepackage[noend]{algpseudocode}

\makeatletter
\def\BState{\State\hskip-\ALG@thistlm}
\makeatother

\theoremstyle{plain}
\newtheorem{thm}{Theorem}[section]

\theoremstyle{definition}
\newtheorem{defn}{Definition}[section]

\newtheorem*{remark}{Remark}

\title{SARSA(0) Reinforcement Learning over Fully Homomorphic Encryption}

\author{Jihoon Suh${}^{1\dagger}$ and Takashi Tanaka${}^{1}$}
\speaker{Jihoon Suh}

\affils{${}^{1}$Department of Aerospace Engineering and Engineering Mechanics, University of Texas at Austin, USA\\
(E-mail: jihoonsuh@utexas.edu, ttanaka@utexas.edu)\\
}

\abstract{%
We consider a cloud-based control architecture in which the local plants outsource the control synthesis task to the cloud. In particular, we consider a cloud-based reinforcement learning (RL), where updating the value function is outsourced to the cloud. To achieve confidentiality, we implement computations over Fully Homomorphic Encryption (FHE). We use a CKKS encryption scheme and a modified SARSA(0) reinforcement learning to incorporate the encryption-induced delays. We then give a convergence result for the delayed updated rule of SARSA(0) with a blocking mechanism. We finally present a numerical demonstration via implementing on a classical pole-balancing problem. }


\usepackage{url}

\begin{document}

\maketitle

%-----------------------------------------------------------------------

\section{Introduction}
Certain control algorithms such as Model Predictive Control (MPC) and visual servoing require heavy real-time computational operations while on-site controllers in industrial control systems are often resource-constrained. The cloud-based control may resolve this issue by allowing on-site controllers to outsource their computations to cloud computers. Many such cloud-based control schemes have been proposed recently in the literature. For instance, \cite{hegazy2014industrial} considered a cloud-based MPC architecture with an application to large-scale power plant operations. \cite{wu2012cloud} considered a cloud-based visual servoing architecture, where an UDP-based communication protocol was developed for latency reduction. 

Despite the computational advantages, a naive implementation of cloud-based control can leak operational records of clients’ control systems, which often contain sensitive information. Since private information can be a critical asset in modern society, an appropriate privacy protecting mechanism is an essential requirement for cloud-based control services in practice. Both non-cryptographic (e.g., differential privacy) and cryptographic (e.g., homomorphic encryption (HE), the focus of this paper) approaches have been considered in the literature. 

Cloud-based control over HE is pioneered by \cite{Kogiso2015CybersecurityEO}, \cite{FAROKHI2016163}, \cite{KIM2016175}, followed by rapid developments of related technologies in more recent works. The primary focus of the encrypted control literature to date has been on the cloud-based \emph{implementations} of control policies. As seen in Fig.~\ref{fig:1} (a), the role of the cloud is simply to map the current state to the control input based on the synthesized control map. In this paper, we consider a cloud-based control \emph{synthesis} problem shown in Fig.~\ref{fig:1} (b). In this case, the role of the cloud is to synthesize a control map $\pi_t$. The implementation of the control policy is done locally by the client. 

We can classify the control synthesis problems into model-based and data-driven approaches. In model-based approaches, a mathematical model of the system is explicitly used for the policy construction (e.g., synthesis of explicit MPC laws \cite{bemporad2002explicit}), while data-driven approaches construct a policy $\pi_t$ from the historical input-output data without using system models (e.g., reinforcement learning by deep Q-network \cite{mnih2015human}). Control synthesis problem usually demands heavy computations or large data requiring privacy protection while evaluating the synthesized already-synthesized map can be trivial. It is thus natural to consider cloud-based control for the synthesis problem.

\begin{figure}
\begin{center}
\includegraphics[width=\columnwidth]{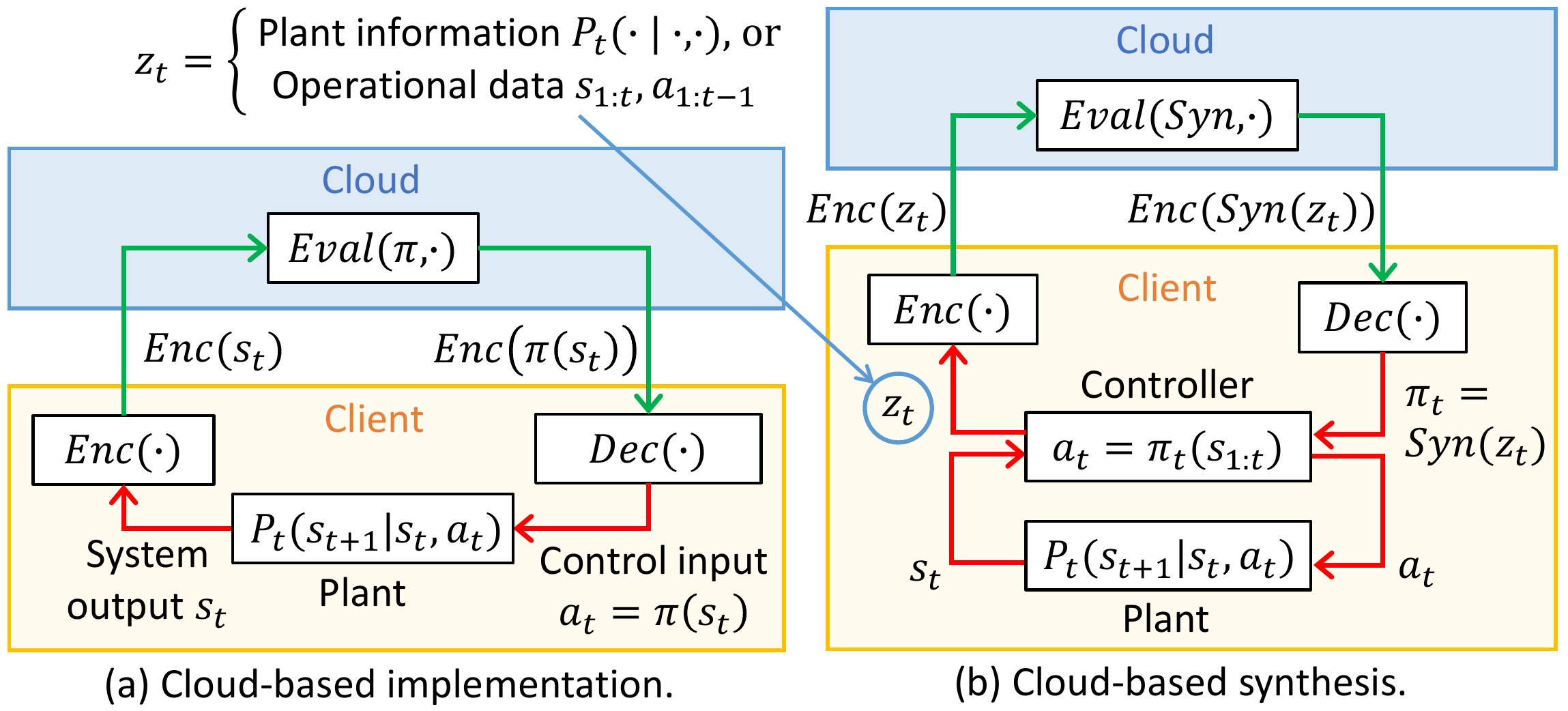}  
\caption{Cloud-based implementation of control policies vs. cloud-based control synthesis. Control synthesis can be based on plant models (model-based approach) or operational data (data-driven approach). } 
\label{fig:1}
\end{center}
\end{figure}

A large portion of the existing work on encrypted control is restricted to partially homomorphic encryption (PHE) schemes. However, FHE is preferred for control synthesis problem because it require more advanced algebraic operations where as PHE can only allow either addition or multiplication. As a step forward to more general control synthesis over FHE, this paper focuses on implementing an elementary reinforcement learning (RL) algorithm, namely the SARSA(0). We note that RL over FHE is largely unexplored, while supervised learning over FHE has been studied in recent years \cite{dowlin2016cryptonets}. 

The use of FHE in encrypted control introduces new challenges such as finite computation span \cite{KIM2016175}, encryption noise, or encryption delay. This paper focuses on the encryption-induced delay due to computational overhead by FHE. We propose a modified SARSA(0) under blocking states and discuss its convergence properties. We implement the modified SARSA(0) over Cheon-Kim-Kim-Song (CKKS) encryption scheme and apply it to a classical RL problem (pole balancing) to demonstrate the feasibility of our study.

The structure of the paper is as follows. In Section 2, we review the basics of HE and overview the encrypted control literature. In Section 3, we review the basics of RL and propose a modified SARSA(0) algorithm. In Section 4, we present the implementation results of a modified SARSA(0) over FHE. We conclude in Section 5 with future research directions.

\section{Preliminaries}
\textit{Homomorphic Encryption} (HE) allows addition and/or multiplcation operation over ciphertexts. A public-key homomorphic encryption method can enable cloud outsourced computations with data protection. A comprehensive introduction to this technology can be found otherwise in \cite{yi2014homomorphic}.
Homomorphic Encryption is \textit{Partially Homomorphic} if only one opeartion (either an addition or a multiplication) is preserved. Among many PHE schemes, ElGamal Cryptosystem and Paillier Cryptosystem have been extensively studied in encrypted control literature. The former is homomorphic with respect to multiplication and the latter with addition, \cite{Kogiso2015CybersecurityEO}, \cite{FAROKHI2016163}, \cite{8619835}, and \cite{8126799}.

If both additions and multiplications are preserved but for a limited number of operations, then the encryption scheme is called \textit{Somewhat (Leveled) Homomorphic}. \cite{homenc} proved that any Somewhat Homomorphic scheme with a \textit{bootstrapping} procedure can be promoted to \textit{Fully Homomorphic Encryption} by controlling the noise growth in ciphertexts. Since then, many more efficient schemes emerged. It is also worth mentioning that many of recent generation FHE schemes are among a few candidates for \textit{Quantum-resistant} cryptosystem, \cite{Lange2015InitialRO}. A more in-depth coverage on FHE can be found in \cite{Halevi2017}.

\subsection{CKKS Encryption scheme}
First developed in \cite{cryptoeprint:2016:421} with a bootstrapping procedure \cite{10.1007/978-3-319-78381-9_14}, CKKS scheme is one of the most prominent fully homommorphic encryption scheme in practice. CKKS is unique in that it can encrypt a vector of complex numbers and can perform approximate arithmetic as opposed to other schemes, which encrypt integers and perform exact arithmetic. Note that a special encoding structure is necessary as the plaintext space of CKKS scheme is an integer-coefficient polynomial ring modulo cyclotomic polynomial. A detailed review of CKKS scheme is out of scope of this paper and we refer \cite{cryptoeprint:2016:421}.

CKKS scheme's precision loss is comparable to a precision loss occuring in unencrypted floating point arithmetic. Therefore, this scheme is very convenient for control applications. It is also a batch-encryption scheme, making it suitable for many data-driven applications. Moreover, SEAL \cite{sealcrypto} provides a very accessbile open-source library.

\subsection{Literature Overview}
More recently, the secure evaluation of affine control law for explicit MPC using Paillier Cryptosystem (PCS) is shown in \cite{8126799}. In \cite{8619835}, an implicit MPC control evaluation using PCS was shown to be possible through the projected Fast Gradient Method and the use of a communication protocol. Also, in \cite{DARUP2018535}, a PCS encrypted implicit MPC was shown via the use of real-time proximal gradient method. But due to the nature of PCS, or any other partially homomorphic schemes, some parameters are assumed to be public. This is a valid assumption in some cases but may not fit some other scenarios.

On the other hand, FHE is not as adopted as PHE in control systems because it is still far more computationally demanding compared to PHE. Nonetheless, the feasibility of using the FHE for a cloud-based control system was first shown in \cite{KIM2016175}. Subsequently, in \cite{kim2019comprehensive}, a secure dynamic control was proposed using LWE-based FHE and a critical observation was that the noise growth of FHE is bounded by the stability of the closed-loop system under some conditions, eliminating the need for bootstrapping, which is one of the most computationally involved procedure.

\section{SARSA(0) over fully homomorphic encryption}

Consider a Markov decision process (MDP) defined by a tuple $(\mathcal{S}, \mathcal{A}, P, r)$, where $\mathcal{S}$ is a finite state space, $\mathcal{A}$ is a finite action space, $P(s_{t+1}|s_t, a_t)$ is the state transition probability and $r(s_t, a_t)$ is the reward. 
A policy is a sequence of stochastic kernels $\pi=(\pi_1, \pi_2, ...)$, where $\pi_t(a_t|s_{1:t}, a_{1:t-1})$ assigns the probability of selecting the next action $a_t\in\mathcal{A}$ given the history $(s_{1:t}, a_{1:t-1})$.
For a fixed policy $\pi$, the \emph{value} of each state $s\in\mathcal{S}$ is defined by
\[
V^\pi(s)={\mathbb{E}}^\pi \left[\sum_{t=1}^\infty\gamma^t r(s_t, a_t)| s_0=s\right],
\]
where $0\leq \gamma <1$ is a predefined discount factor. We say that a policy $\pi^*$ is \emph{optimal} if it maximizes the value of each state simultaneously. 
The existence of time-invariant, Markov, and deterministic optimal policy under the present setup is well-known, (e.g., \cite{bertsekas2011dynamic}).
Consequently, an optimal policy of the form $a_t=\pi^*(s_t)$ can be assumed without loss of performance. 
The value function $V^*$ under an optimal policy $\pi^*$ satisfies the Bellman's equation 
\begin{equation}
\label{eq:bellman_v}
V^*(s)=\max_{a\in\mathcal{A}} \left[r(s,a)+\gamma \sum_{s'\in\mathcal{S}}P(s'|s,a)V^*(s')\right].
\end{equation}

\subsection{Q-learning}
The focus of reinforcement learning algorithms in general is to construct an optimal policy $\pi^*$ based on the observed history of the states, actions, and reward signals. 
The Q-learning, \cite{watkins1992q}, achieves this by attempting the construction of the optimal \emph{Q-function} defined by
\begin{equation}
\label{eq:def_q}
Q^*(s,a):=r(s,a)+\gamma\sum_{s'\in\mathcal{S}}P(s'|s,a)V^*(s')
\end{equation}
for each $(s, a)\in\mathcal{S}\times\mathcal{A}$. 
It follows from the Bellman's equation \eqref{eq:bellman_v} that
\[
Q^*(s,a)=r(s,a)+\gamma\sum_{s'\in\mathcal{S}}P(s'|s,a)\max_{a'\in\mathcal{A}} Q^*(s',a')
\]
for each $(s, a)\in\mathcal{S}\times\mathcal{A}$. Once $Q^*$ is obtained, an optimal policy can also be obtained by $\pi^*(s)=\argmax_{a\in\mathcal{A}} Q^*(s,a)$.

Let $\{\alpha_n\}_{n=0,1,2,...}$ be a predefined sequence such that $0\leq \alpha_n\leq 1 \; \forall n$, $\sum_{n=0}^\infty \alpha_n=\infty$ and $\sum_{n=0}^\infty\alpha_n^2<\infty$. 
Denote by $n_t(s,a)$ the number of times that the state-action pair $(s,a)$ has been visited prior to the time step $t$.\footnote{If $(s,a)$ is visited for the first time at $t=1$, then $n_1(s,a)=0$ and $n_2(s,a)=1$.} 
Upon the observation of $(s_t, a_t, r_t, s_{t+1})$, the Q-learning updates the $(s_t, a_t)$ entry of the Q-function by
\begin{align}
Q_{t}&(s_t,a_t)=(1-\alpha_{n_t(s_t,a_t)})Q_{t-1}(s_t,a_t) \nonumber \\
&+\alpha_{n_t(s_t,a_t)}\left(r_t+\gamma \max_{a'\in\mathcal{A}}Q_{t-1}(s_{t+1},a')\right). \label{eq:q-learning}
\end{align}
No update is made to the $(s,a)$ entry if $(s,a)\neq(s_t, a_t)$.
The following result is standard in the literature:
\begin{thm}
\label{thm_q_learning}
For an arbitrarily chosen initial Q-function $Q_0$, the convergence $Q_t(s,a)\rightarrow Q^*(s,a)$ as $t\rightarrow \infty$ holds under the update rule \eqref{eq:q-learning} provided each state-action pair $(s,a)\in\mathcal{S}\times\mathcal{A}$ is visited infinitely often.
\end{thm}
In what follows, we assume that the underlying MDP is communicating, i.e., every state can by reached from every other state under certain policies.

\subsection{SARSA(0)}
The Q-learning method described above is considered an off-policy reinforcement learning algorithm in the sense that the hypothetical action $a'$ in the update rule \eqref{eq:q-learning} need not be actually selected as $a_{t+1}$. The SARSA(0) algorithm, on the other hand, is an on-policy counterpart that performs the Q-function update based on the experienced trajectory $(s_t, a_t, r_t, s_{t+1}, a_{t+1})$:
\begin{align}
Q_t(s_t,a_t)=&(1-\alpha_{n_t(s_t,a_t)})Q_{t-1}(s_t,a_t) \nonumber \\
&+\alpha_{n_t(s_t,a_t)}\left(r_t+\gamma Q_{t-1}(s_{t+1},a_{t+1})\right). \label{eq:sarsa}
\end{align}
We emphasize that the absence of the \emph{max} operation in \eqref{eq:sarsa} is a significant advantage for the implementation over FHE because the current polynomial approximations of comparison operations are highly inefficient, \cite{cryptoeprint:2019:1234}.

The convergence of SARSA(0) can be guaranteed under certain conditions.
Although a complete discussion must be differed to \cite{singh2000convergence}, roughly speaking it requires that (i) the learning policy is infinitely exploring, i.e., each state-action pair $(s,a)\in\mathcal{S}\times\mathcal{A}$ is visited infinitely often; and (ii) the learning policy is greedy in the limit, i.e., the probability that $a_{t+1}\neq \argmax_{a'\in\mathcal{A}}Q_{t-1}(s_{t+1}, a')$ tends to zero.
(Condition (i) is required for the convergence of Q-learning (off-policy counterpart of SARSA). The additional condition (ii) is needed due to the on-policy nature of SARSA(0). 
The following is a simple example of learning policies satisfying (i) and (ii):
\begin{defn}
\label{defn:epsilon}
(Decreasing $\epsilon$ policy) Let $n_t(s)$ be the number of times that the state $s$ has been visited prior to time step $t$ and define $\epsilon_t(s)=c/n_t(s)$ for some $0<c<1$. We say that $\pi$ is a decreasing $\epsilon$ policy if it selects an action $a_{t+1}$ randomly with the uniform distribution over $\mathcal{A}$ with probability $\epsilon_{t+1}(s_{t+1})$ and the greedy action $\argmax_{a\in\mathcal{A}}Q_{t-1}(s_{t+1}, a)$ with probability $1-\epsilon_{t+1}(s_{t+1})$.
\end{defn}
The following result is a consequence of Theorem 1 in \cite{singh2000convergence}.
\begin{thm}
\label{thm:sarsa}
For an arbitrarily chosen $Q_0$, the convergence $Q_t(s,a)\rightarrow Q^*(s,a)$ as $t\rightarrow \infty$ for each state-action pair $(s,a)\in\mathcal{S}\times\mathcal{A}$  occurs with probability one under the SARSA(0) update rule \eqref{eq:sarsa} and the decreasing $\epsilon$ policy.
\end{thm}
\begin{proof}
(Outline only) The result relies on the convergence of stochastic approximation \cite[Lemma 1]{singh2000convergence}, whose premises are satisfied if the learning policy is \emph{Greedy in the limit with infinite exploration (GLIE)}: (i) infinitely exploring, and (ii) greedy in the limit in that ${\mathbb{E}}|Q_{t-1}(s_{t+1}, a_{t+1})-\max_{a\in\mathcal{A}}Q_{t-1}(s_{t+1},a)|\rightarrow 0$ as $t\rightarrow \infty$ with probability one. To verify (i), we note that, under the assumption of communicating MDP, performing each action \emph{infinitely often} in each state is sufficient to guarantee the infinite exploration of states. Let $t_s(i)$ be the time step that the state $s$ is visited the $i$-th time. Since we are adopting the decreasing $\epsilon$ policy, e.g.,  $\epsilon_t(s) = c/n_t(s)$ with  0 $<$ c $<$ 1,  $\text{Pr}(a_{t_s(j)}=a)$ is dependent on $\text{Pr}(a_{t_s(i)}=a)$. But by the extended Borel-Cantelli lemma, we have w.p.1  $\sum_{i=1}^\infty \text{Pr}(a_{t_s(i)}=a) =\infty$ for each $a\in\mathcal{A}$. Thus, by Lemma 4 (Singh, 2000), we have $n_t(s,a) \rightarrow \infty$ a.s., where $n_t(s,a)$ denotes the number of actions performed at state $s$ at time $t$.
The condition (ii) holds by construction of the decreasing $\epsilon$-greedy policy. The conditions 1, 2, and 3 of the Theorem 1 in \cite{singh2000convergence} are satisfied by construction.

Parameters $\alpha_{n_t}(s_t, a_t)$ and $\gamma$ are often called learning rate and discount rate, respectively. The former is responsible for the convergence and its rate and is analogous to the step size in stochastic gradient descent. The discount rate determines the weight on the future rewards; a small discount rate corresponds to weighing the immediate rewards much higher. These parameters are more than often determined empirically because a single number does not work for different models and algorithms.
\end{proof}

\begin{figure}
\begin{center}
\includegraphics[width=8.1cm]{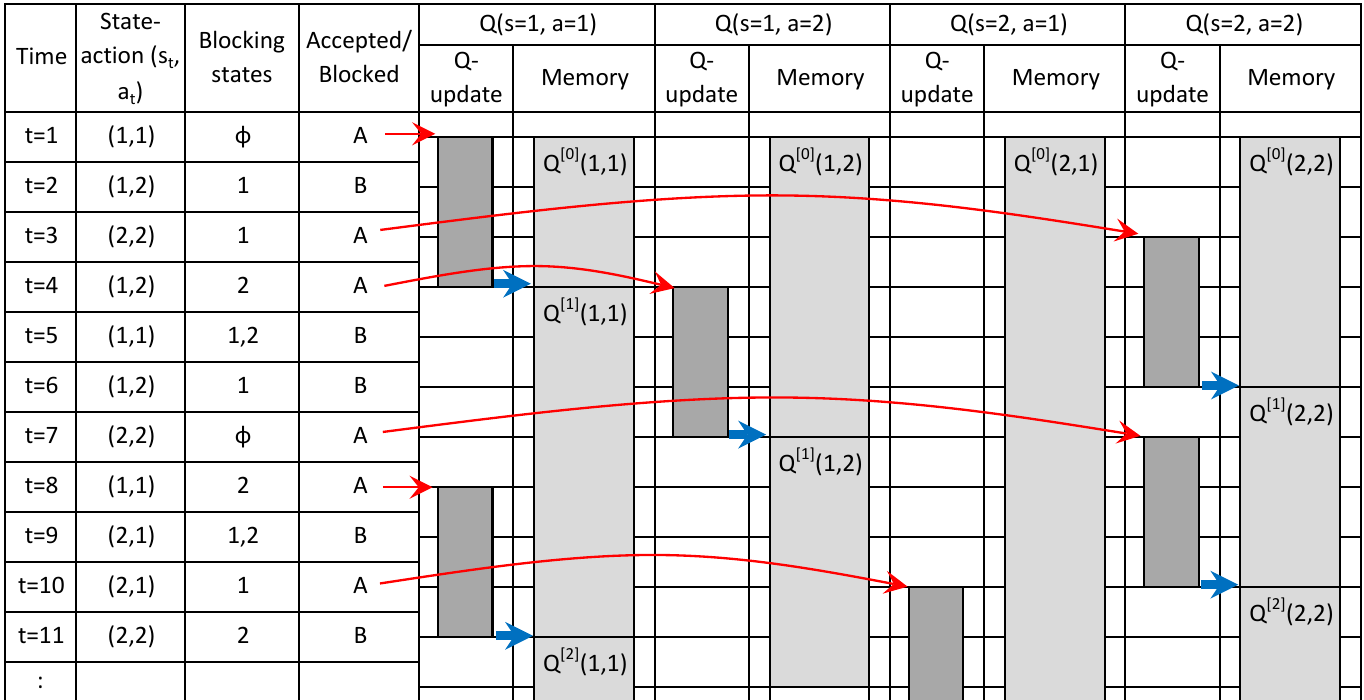}    % The printed column width is 8.4 cm.
\caption{Q-table update with blocking states.} 
\label{fig:5}
\end{center}
\end{figure}

\subsection{SARSA(0) with blocking states}

Consider now a cloud-based implementation of SARSA(0). 
We assume that the Q-function update \eqref{eq:sarsa} is performed by the cloud while the decreasing $\epsilon$ policy is implemented by the client.
As shown in Fig.~\ref{fig:1}(b), at each time step $t$, the client can encrypt and upload a new data set $z_t=(s_t, a_t, r_t, s_{t+1}, a_{t+1})$ and, upon the completion of the Q-function update \eqref{eq:sarsa} on the cloud, the updated entry of the Q-function is downloaded and decrypted. If the computation of Q-update takes less than a unit time interval, then SARSA(0) together with the decreasing $\epsilon$ policy as described above can be implemented in the considered cloud-based architecture without any modifications. 

\begin{remark}
As the name suggests, the data set for tabular SARSA(0) is $z_t=(s_t, a_t, r_t, s_{t+1}, a_{t+1})$ when implemented locally. However, in the cloud implementation considered in \eqref{eq:sarsa}, the data $s_t, s_{t+1}, a_t, a_{t+1}$ only serve the purpose of indexing the correct entries within the copy of $Q$ table at the cloud. Thus, the data communicated is $z_t=(Q(s_t, a_t), r_t, Q(s_{t+1}, a_{t+1})$.
\end{remark}

However, encrypted Q-update can take up to some $L(\geq 2)$ unit time intervals. In such scenarios, a new data set may arrive before the previous Q-update is complete. For simplicity, we do not consider an analysis of communication delay between the cloud and the client. In what follows, we propose a modified SARSA, which we call \emph{SARSA(0) with blocking states}. The proposed Q-update rule is depicted in Fig.~\ref{fig:5}, where a special case with $|\mathcal{S}|=|\mathcal{A}|=2$ and $L=3$ is shown.

First, encrypted values of the initial Q-function
\[
Q_0=(Q^{[0]}(1,1), Q^{[0]}(1,2), Q^{[0]}(2,1), Q^{[0]}(2,2))
\]
are recorded on the cloud's memory. 
At time step $t=1$, the state-action pair $(s_1, a_1)=(1,1)$ is visited, so the encrypted data set $z_1=(Q(s_1,a_1), r_1, Q(s_2, a_2))$ is uploaded to the cloud. 
The computation \eqref{eq:sarsa} for $Q(1,1)$-update is initiated, which will take three time steps to complete. While the $Q(1,1)$ is being updated on the next $L-1$ time steps (i.e., $t=2,3$), the corresponding state $s=1$ is added to the list of \emph{blocking states}. When the state $s$ is blocking, no updates are allowed to the entries of $Q(s, \cdot)$.
For instance, at $t=2$, the state-action pair $(s_2, a_2)=(1,2)$ is visited.
However, since the state $s=1$ is in a blocking state list, the update request is rejected and the data set $z_2=(Q(s_2, a_2), r_2, Q(s_3, a_3))$ is discarded.
At $t=3$, the state-action pair $(2,2)$ is visited. Since $s=2$ is not blocking at $t=3$, the update is accepted, and the computation of $Q(2,2)$ update is initiated using the new data set $z_3=(Q(s_3, a_3), r_3, Q(s_4, a_4))$. At $t=4$, the computation for $Q(1,1)$ is complete and the result $Q^{[1]}(1,1)$ is recorded on the memory. ($Q^{[k]}(i,u)$ meaning the $k$-th revision of the $Q(i,u)$.) Since the state-action pair $(s_4, a_4)=(1,2)$ is visited when the state $s=1$ is removed from the blocking list, a new data set $z_4$ is accepted and the computation for $Q(1,2)$-update is initiated.

Under SARSA(0) with blocking states described above, we denote by $Q_t$ the most updated version of the Q-function recorded on the memory as of time step $t$. For instance, Fig.~\ref{fig:5} reads
\[
Q_{10}=(Q^{[1]}(1,1), Q^{[1]}(1,2), Q^{[0]}(2,1), Q^{[2]}(2,2)).
\]

Let $m_t = m_t(s_t,a_t)$ be the number of times that the update has been accepted at the state-action pair $(s,a)$ prior to time step $t$.
Whenever the update with data $z_{t-L}=(Q(s_{t-L}, a_{t-L}), r_{t-L}, Q(s_{t-L+1}, a_{t-L+1}))$ is accepted at time step $t-L$, the following update is made at time step $t$:
\begin{align*}
Q_t&(s_{t-L},a_{t-L})=(1-\alpha_{m_{t-L}})Q_{t-L}(s_{t-L},a_{t-L}) \\
&+\alpha_{m_{t-L}}\left(r_{t-L}+\gamma Q_{t-L}(s_{t-L+1},a_{t-L+1})\right). 
\end{align*}
Since the blocking mechanism ensures that the entries $Q_k(s_{t-L},\cdot)$ are unchanged over the time steps $t-L\leq k \leq t-1$, the above update rule can also be written as
\begin{align*}
Q_t&(s_{t-L},a_{t-L})=(1-\alpha_{m_{t-L}})Q_{t-1}(s_{t-L},a_{t-L}) \\
&+\alpha_{m_{t-L}}\left(r_{t-L}+\gamma Q_{t-1}(s_{t-L+1},a_{t-L+1})\right). 
\end{align*}

\begin{defn}
\label{defn:epsilon_delay}
We define the \emph{decreasing $\epsilon$ policy with delay} similarly to the decreasing $\epsilon$ policy (Definition~\ref{defn:epsilon}) except that greedy actions are selected by $a_{t+1}=\argmax_{a\in\mathcal{A}}Q_{t-1}(s_{t+1},a)$, where $Q_{t-1}$ is the most updated version of the Q-function available at $t-1$.
\end{defn}

We have the following convergence result.
\begin{thm}
For an arbitrarily chosen $Q_0$, the convergence $Q_t(s,a)\rightarrow Q^*(s,a)$ as $t\rightarrow \infty$ for each state-action pair $(s,a)\in\mathcal{S}\times\mathcal{A}$ occurs with probability one under SARSA(0) with delayed update and blocking if the decreasing $\epsilon$ policy with delay is adopted.
\end{thm}
\begin{proof}
We will prove that, under delayed update and blocking mechanism, (i) every state is explored infinitely often, and (ii) the update is still executed infinitely often, and (iii) the policy is greedy in the limit, again satisfying the convergence conditions.

(i) Since the decreasing $\epsilon$ policies with and without delay (Definitions~\ref{defn:epsilon} and \ref{defn:epsilon_delay}) share the same probability $\epsilon_t(s)$ for random exploration, the inequality $\sum_{i=1}^\infty \text{Pr}(a_{t_s(i)}=a)\geq \sum_{i=1}^\infty c/i=\infty$ remains valid under the decreasing $\epsilon$ policy with delay. As in the proof of Theorem~\ref{thm:sarsa}, this implies that each state-action pair is visited infinitely often. 

(ii) Assume that update occurs only $K(<\infty)$ times for some state-action pair $(s,a)$. Since $L<\infty$, the Q-update can be blocked at most $K(L-1)(<\infty)$ times. This means that updates are accepted infinitely many times, and this contradicts the assumption that updates occur only $K$ times.

(iii) The policy is greedy in the limit in that 
\begin{equation}
\label{eq:sarsa_delay_pf}
{\mathbb{E}}|Q_{t-1}(s_{t-L+1}, a_{t-L+1})-\max_{a\in\mathcal{A}}Q_{t-1}(s_{t-L+1},a)|\rightarrow 0
\end{equation}
as $t\rightarrow \infty$ with probability one.
To see this, suppose that $a_{t-L+1}$ was a greedy action selected at time step $t-L$, i.e.,
\[
a_{t-L+1}=\argmax_{a\in\mathcal{A}}Q_{t-L}(s_{t-L+1}, a).
\]
Due to the blocking mechanism, $Q_k(s_{t-L},\cdot)$ are unchanged over the time steps $t-L\leq k \leq t-1$. This means that 
\[
a_{t-L+1}=\argmax_{a\in\mathcal{A}}Q_{t-1}(s_{t-L+1}, a).
\]
Thus, whenever $a_{t-L+1}$ was selected greedily at $t-L$, we have
\[
Q_{t-1}(s_{t-L+1}, a_{t-L+1})-\max_{a\in\mathcal{A}}Q_{t-1}(s_{t-L+1},a))=0.
\]
Since the delayed $\epsilon$-greedy policy is greedy in the limit, the convergence \eqref{eq:sarsa_delay_pf} holds with probability one.

\end{proof}

\section{Numerical demonstration}
We used a classical pole-balancing problem from RL. The setup can be found from \cite{sutton_barto_2018}. Model parameters are unchanged but SARSA(0) updates are done over ciphertexts. For encryption, we used Microsoft SEAL to set up the CKKS scheme. The step size $\alpha_{m_t(s_t,a_t)} = \alpha$ is a constant but the policy parameter $\epsilon$ is decreased as the $Q$ values are updated. At every $L$ step, the client transmits Q-values of the state-action pairs it visited directly to the cloud.

Algorithm 1 shows the pseudocode of Encrypted SARSA(0) with delayed updates.  The data set protected here is $z_i = (Q_{i}(s_i,a_i), r_i, Q_{i}(s_{i+1}, a_{i+1}), \alpha, \gamma)$. For simplicity in notation, we use  $Q_{i}'$ to denote $Q_{i}(s_{i+1}, a_{i+1})$ in Algorithm 1. To take advantage of the batch encryption, the agent is allowed to explore many state-action pairs before updating $Q$ values at the cloud. The agent then uploads the batch data of $z_{i}$ and the cloud sequentially updates the $Q$ values. We use a tilde to denote the encoded version (e.g., $\tilde{z}$ to be the encoded version of $L$ data sets, $\tilde{z} = Encode([z_{t-L}, \dots, z_{t-1}])$) and a symbol $\#$ in front to denote the encrypted version (e.g., $\#\tilde{z}$ is the encrypted version, $\#\tilde{z} = Encrypt(\tilde{z})$).

\begin{algorithm}
\caption{Encrypted SARSA(0)}\label{euclid}

\textbf{Client (\emph{Start})}
\begin{algorithmic}[1]
\State Perform actions and state transitions $L$ times according to the up-to-date Q-table values
\State Store data $z_i$, where $i = t-L, \dots, t-1$
\State Encode $[z_{t-L}, \dots, z_{t-1}]$ and get $\tilde{z}$
\State Encrypt and get $\#\tilde{z}$ and upload to the cloud
\end{algorithmic}

\textbf{Cloud}
\begin{algorithmic}[1]
\State Update: $\#\tilde{Q} \gets (1-\#\tilde{\alpha})\#\tilde{Q} + \#\tilde{\alpha}(\#\tilde{r} + \#\tilde{\gamma}\#\tilde{Q^{\prime}})$ 
\State Upload the updated $\#\tilde{Q}$ back to the Client.
\end{algorithmic}

\textbf{Client}
\begin{algorithmic}[1]
\State Decrypt $\#\tilde{Q}$ to get $\tilde{Q}$
\State Decode $\tilde{Q}$ to get the updated vector of $Q_i$'s
\State \textbf{go to} \emph{Start}
\end{algorithmic}
\end{algorithm}

Table \ref{tab:table1} lists the CKKS Parameters used for the demonstration. Table \ref{tab:table2} shows the delay introduced by operations involved in Homomorphic Encryption. The maximum precision error in $Q$ values induced by the encrypted updates was only 0.0063\% which is small enough to be ignored in most practical applications. 
\begin{table}[htb]
  \begin{center}
    \caption{Encryption Parameters}
    \label{tab:table1}
    \begin{tabular}{|c|c|c|c|}
      \textbf{CKKS Parameters} & \textbf{Chosen} \\ % <-- added & and content for each column
      \hline
      \textbf{$N$} & 8192 \\
      \textbf{$q$} & ($2^{50}$, $2^{30}$, $2^{30}$, $2^{30}$, $2^{50}$) \\
      \textbf{Scale Factor} & $2^{30}$ \\
      \textbf{Available Slots} & $4096$

    \end{tabular}
  \end{center}
\end{table}

\begin{table}[htbp]
  \begin{center}
    \caption{Number of operations and average time consumed per each batch update}
    \label{tab:table2}
    \begin{tabular}{|c|c|c|c|}
      \textbf{Type} & \textbf{Num} &\textbf{Time (ms)} & \textbf{Percent}\\ % <-- added & and content for each column
      \hline
      \textbf{Encode} & 5 & 6.695 & 23.59 \%\\ % <--
      \textbf{Encrypt} & 5 &  33.519 & 39.38 \% \\
      \textbf{Multiply} & 4 & 2.549 & 2.99 \% \\ 
      \textbf{Relinearize} & 4 & 14.909& 17.51 \% \\
      \textbf{Rescale} & 4 & 7.886 & 9.26 \% \\
      \textbf{Addition} & 3 & 0.074 & 0.09 \% \\
      \textbf{Decrypt} & 1 & 1.225 & 1.44 \% \\
      \textbf{Decode} & 1 & 4.881 & 5.73 \%

    \end{tabular}
  \end{center}
\end{table}

For this experiment, we used $L = 1000$ to generate the data because the state and action sizes were small for this example. However, the computation time at the cloud is fixed up to $L = 4096$ with the same encryption parameters. Thus it can allow larger batch operations if necessary, such as systems with more number of states and/or actions. As usual in classical SARSA(0), we observed that encrypted $Q$ value updates converged with large iterations and the inverted pendulum simulation showed successful balancing. As expected, increasing $L$ sped up the learning because the cloud could perform the time-consuming encrypted operation in a more compact manner.

We note that the client is still burdened with non-trivial computing tasks, most notably CKKS encoding and homomorphic encryption, which takes up more than a half of the homomorphic operations involved. On the other hand, decryption and decoding tasks are less strenuous. This is a prevalent issue in encrypted control as seen in the results of \cite{DARUP2018535} and many others, dwindling the appeal of cloud-based control. Nevertheless, the privacy guarantee acquired can compensate for increased computing burden and can be suitably applied to problems where control synthesis can be done off-line.

\section{Summary and future work}
We provided a general framework for control \emph{synthesis} over FHE. In particular, we studied the feasibility of an encrypted SARSA(0). We showed a convergence result for the SARSA(0) with delayed updates under a blocking state mechanism. We then applied it on a classical problem of inverted pole balancing. Numerical results showed that the encrypted learning via SARSA(0) induced minimal precision loss. The algorithm still showed convergence under delay but nonetheless it required more computing time than unencrypted version in exchange for the privacy guarantee. 

Since the time scales of control synthesis (i.e., the frequency at which control policies are updated) are typically much slower than that of implementation, we expect that the framework shown in this paper is more suitable for encrypted control. Moreover, sensor data and control command are not encrypted at the plant, allowing fast implementation of the control policy once the policy is synthesized. Thus, if synthesis could be done offline, the method proposed can be applied in practice with no loss from the encryption induced delay.

Although SARSA(0) is an elementary algorithm, it is easier to analyze theoretically than more advanced deep learning algorithms. It can also be the basis of more advanced encrypted learning implementations.
Many challenges remain for encrypted control in general. The delays introduced can limit the area of applications. Suitable area of applications can be slowly-varying systems with critical requirement on privacy breach. The other critical challenge is the difficulty of implementing operations such as comparison and sorting over FHE. This makes the execution of maximum on ciphertext domain very challenging. 

A natural extension to the proposed study is to upgrade the synthesis algorithm to more advanced control synthesis algorithms such as using artificial neural network or model predictive control over FHE. To this end, an efficient polynomial comparison function could be of significant value. Also, only numerical experiment was performed regarding the encryption induced noise. We aim to study analytically the effect of encryption induced noise on the stability of the learning and control algorithms over FHE.

\bibliography{reference}

\end{document}